\documentclass[a4paper,USenglish,cleveref,autoref,thm-restate]{lipics-v2021}

\hideLIPIcs
\nolinenumbers
\usepackage[utf8]{inputenc}
\usepackage{tikz}
\tikzstyle{vertex}=[draw, circle, fill, inner sep = 2.4pt]

\usepackage{xspace}

\renewcommand{\cal}{\mathcal}
\DeclareMathOperator{\sgn}{sgn}
\renewcommand{\t}{\text}

\DeclareMathOperator{\Pf}{Pf}

\DeclareMathOperator{\poly}{poly}

\newcommand{\sub}{\subseteq}
\newcommand{\set}[1]{\left\{#1\right\}}
\newcommand{\grp}[1]{\left(#1\right)}

\renewcommand{\P}{\textsf{P}}
\newcommand{\NP}{\textsf{NP}}
\newcommand{\eps}{\varepsilon}
\newcommand{\FF}{\mathbb{F}}

\newcommand{\Color}{{\normalfont\textsc{Coloring}}}
\newcommand{\DualColor}{{\normalfont\textsc{Dual Coloring}}}
\newcommand{\CliqueCover}{{\normalfont\textsc{Clique Cover}}}
\newcommand{\DualCliqueCover}{{\normalfont\textsc{Dual Clique Cover}}}

\newcommand{\SetCover}{{\normalfont\textsc{Set Cover}}}

\newcommand{\Clique}{\normalfont\textsc{Clique Cover}}

\definecolor{amethyst}{rgb}{0.6, 0.4, 0.8}

\renewcommand{\bar}{\overline}
\newcommand{\mA}{\mathbf{A}}
\newcommand{\mB}{\mathbf{B}}
\newcommand{\stepfont}[1]{\textsf{\textbf{\textcolor{lipicsGray}{#1}}}}
\newcommand{\pair}[1]{\left\langle #1\right\rangle}

\title{Graph Coloring Below Guarantees via Co-Triangle Packing}

\author{Shyan Akmal}{INSAIT, Sofia University ``St. Kliment Ohridski'', Bulgaria \and \url{https://www.shyanakmal.com}}{shyan.akmal@gmail.com}{https://orcid.org/0000-0002-7266-2041}{Partially funded by the Ministry of Education and Science of Bulgaria (support
for INSAIT, part of the Bulgarian National Roadmap for Research Infrastructure).}

\author{Tomohiro Koana}{Keio University, Japan}{tomohiro.koana@gmail.com}{https://orcid.org/0000-0002-8684-0611}{Supported by JST ERATO Grant Number JPMJER2301, Japan.}

\authorrunning{S. Akmal and T. Koana}


\acknowledgements{We thank the anonymous reviewers for helpful feedback on this work.}

\Copyright{Shyan Akmal and Tomohiro Koana}


\begin{document}

\maketitle

\begin{abstract}

In the $\ell$-\Color{} problem, we are given a graph on $n$ nodes,
and tasked with determining if its vertices can be properly colored using $\ell$ colors.
In this paper we study 
\emph{below-guarantee graph coloring}, which tests whether an $n$-vertex graph can be properly colored using $g-k$ colors, where $g$ is a trivial upper bound such as~$n$.  
We introduce an algorithmic framework that builds on a packing of co-triangles $\overline{K_3}$ (independent sets of three vertices):  
the algorithm greedily finds co-triangles and employs a win-win analysis.  
If many are found, we immediately return \emph{yes};  
otherwise these co-triangles form a small \emph{co-triangle modulator}, whose deletion makes the graph co-triangle-free.

Extending the work of [Gutin et al., SIDMA 2021], who solved $\ell$-\Color{} (for any $\ell$) in randomized $O^{\ast}(2^{k})$ time when given a $\overline{K_2}$-free modulator of size~$k$, we show that this problem can likewise be solved in randomized $O^*(2^{k})$ time when given a $\overline{K_3}$-free modulator of size~$k$.

This result in turn yields a randomized $O^{*}(2^{3k/2})$ algorithm for $(n-k)$-\textsc{Coloring} (also known as \textsc{Dual Coloring}), improving the previous $O^{*}(4^{k})$ bound.
We then introduce a smaller parameterization, $(\omega+\overline{\mu}-k)$-\textsc{Coloring}, where $\omega$ is the clique number and $\overline{\mu}$ is the size of a maximum matching in the complement graph; since $\omega+\overline{\mu}\le n$ for any graph, this problem is strictly harder.  
Using the same co-triangle-packing argument, we obtain a randomized $O^{*}(2^{6k})$ algorithm, establishing its  fixed-parameter tractability for a smaller parameter.  
Complementing this finding, we show that no fixed-parameter tractable algorithm exists for $(\omega-k)$-\textsc{Coloring} or $(\overline{\mu}-k)$-\textsc{Coloring} under standard complexity assumptions.

\end{abstract}

\newpage 

\section{Introduction}
\label{sec:intro}

Graph coloring is a cornerstone of theoretical computer science and discrete mathematics: it connects to deep structure theorems and has driven the development of numerous algorithmic techniques \cite{DBLP:journals/ijm/AppelH77,BHK2009,DBLP:journals/ejc/BollobasCE80,DBLP:journals/anm/ChudnovskyRST06,DBLP:books/wi/JensenT11}.
Given an undirected graph $G$ on $n$ vertices,
a \emph{proper coloring} of $G$ is an assignment of colors to its nodes such that any two adjacent vertices receive different colors. 
The canonical decision problem for this concept is $k$-\Color{}: given an integer $k$ and a graph $G$, decide whether $G$ admits a proper coloring with at most $k$ colors.

While 2-\Color{} is polynomial-time solvable, 3-\Color{} (and therefore $k$-\Color{} for every $k\ge3$) is already  \NP{}-hard \cite{Lovasz1973}.
The field of parameterized complexity tackles such \NP{}-hard tasks by measuring the complexity of solving these problems in terms of secondary measures.
In this paper we focus on a particularly important subarea of this field 
known as \emph{below-guarantee parameterization}.

\subparagraph*{Below-Guarantee Parameterization.}
In the guarantee parameterization paradigm,
we reformulate problems by first identifying classes of instances that are trivially solvable---typically because a solution is guaranteed to exist---and then parameterize general instances by their distance from these ``easy'' cases.
See  \cite{GutinMnichSurvey} for a survey of this area.

\subparagraph{Dual Coloring.}
In the context of coloring graphs with $n$ nodes,
observe that the $n$-\Color{} problem is trivial, 
because 
we can obtain a proper coloring by assigning each vertex a distinct color.
The framework
of below-guarantee parameterization described above then motivates studying the $(n-k)$-\Color{} problem, also referred to as \DualColor{}, 
where the parameter $k$ captures how many colors we are trying to ``save'' compared to the trivial coloring using different colors for all vertices. 
The \DualColor{} problem 
has been influential in parameterized complexity, with its study in \cite{ChorFJ04}
introducing the concept of crown reductions, by now a basic technique in the field of kernelization \cite[Chapter~4]{KernelizationTextbook}.  
Given an  $n$-node instance $G$ of 
\DualColor{} and any constant $\eps > 0$, \cite{LDYR2021} showed that in polynomial time one can kernelize the instance to a graph $G'$ on $(2+\eps)k$ vertices
such that the answer to the \DualColor{} problem is the same on $G$ and $G'$, for sufficiently large $k$ in terms of $\eps > 0$.
Combining this kernelization with 
existing $2^n\poly(n)$ time algorithms for graph coloring \cite{BHK2009},
we immediately recover a $O^*(4^{(1+\eps)k})$ time algorithm for \DualColor{} for any constant $\eps > 0$,
where we write $O^*(f(k))$ to denote $f(k)\poly(n)$. 
Without going through this kernelization, one can also use the algorithms of \cite{GMOW2021} to solve \DualColor{} directly in $O^*(4^k)$ time.
We improve upon this result:

\begin{restatable}{theorem}{dualcolor}
\label{thm:dualcolor}
    \DualColor{} can be solved in randomized $O^*(2^{3k/2})$ time.
\end{restatable}

Let us mention a closely related problem called $(n-k)$-\SetCover{}. 
In this problem, we are given a family $\mathcal{F}$ of subsets of $[n]$, and are tasked with determining if there exist $(n - k)$ sets from $\mathcal{F}$ whose union is $[n]$. 
Very recently, Alferov et al.~\cite{AlferovBB24} showed that
$(n-k)$-\SetCover{} can be solved in $O^*(2^{3k/2}|\mathcal{F}|)$ time. 
Even though the $(n-k)$-\Color{} problem can be cast as an instance of $(n-k)$-\SetCover{}, where $\mathcal{F}$ is the collection of all independent vertex sets of the graph,
this result of \cite{AlferovBB24} does not imply an $O^*(2^{3k/2})$ time algorithm for $(n-k)$-\Color{},
because $\mathcal{F}$ can have exponentially many sets in $n$. 
The dynamic programming approach used in the algorithm of \cite{AlferovBB24} appears to inherently require this runtime dependence on $|\mathcal{F}|$ even when the sets $\mathcal{F}$ can be implicitly described as in $(n-k)$-\Color{},
so that \Cref{thm:dualcolor} is not directly implied by this previous work. 
However, we note it seems plausible that one could combine the ideas of \cite{AlferovBB24} with  the subset convolution arguments of~\cite{subset-convolution} to obtain fast algorithms for \DualColor{} as well, and recover \Cref{thm:dualcolor}.

\subparagraph*{A Stronger Structural Guarantee.}  
Given a graph $G$, we let $\overline{G}$ denote its complement graph, which has the same vertex set as $G$, but an edge $\set{u,v}$ if and only if $u$ and $v$ are distinct, non-adjacent nodes in $G$. 
Let $\omega = \omega(G)$ denote the size of the maximum clique in $G$.
Let $\bar{\mu} = \mu(\overline{G})$ denote the size of the maximum matching in the complement $\overline{G}$.
Then we have the following structural parameter-based guarantee for coloring $G$.

\begin{observation}
    \label{obs:structure-guarantee}
    The graph $G$ admits a proper coloring using at most $\grp{\omega + \bar{\mu}}$ colors. 
\end{observation}
\begin{proof}
    Let $M$ be a maximal matching in $\overline{G}$.
By definition, $|M|\le \bar{\mu}$.
By maximality of $M$ in~$\overline{G}$,
the vertices that are not incident to any edge in $M$ form an independent set in $\bar{G}$.
Thus the set $C$ of vertices outside $M$  form a clique in $G$, so $|C|\le \omega$.
We color $G$ as follows.
For each edge in $M$, we assign a unique color to both of its endpoints, and we assign distinct new colors to all vertices in $C$.
In total, we use at most $(\omega+\bar{\mu})$ colors.
By construction this coloring is proper, which proves the claim. 
\end{proof}

Note that in any graphs with $n$ vertices, we always have  
    $\grp{\omega + \bar{\mu}} \le n.$
This is because if we take a clique $C$ in $G$ of size $\omega$ and a matching $M$ in $\bar{G}$ of size $\bar{\mu}$, then by definition each edge of $M$ is incident to at least one distinct vertex outside of $C$, which forces $n\ge |C| + |M| = \grp{\omega + \bar{\mu}}$.
Thus the guarantee of \Cref{obs:structure-guarantee} is stronger than the trivial fact that any graph with $n$ vertices always has a proper coloring using $n$ colors. 

In the same way this latter observation motivated the \DualColor{} problem,
it is natural to ask if we can obtain efficient algorithms while parameterizing below the stronger guarantee provided by \Cref{obs:structure-guarantee}.
We prove that this is indeed possible.

\begin{restatable}[Parameterization from Stronger Guarantee]{theorem}{guaranteefpt}
\label{thm:guarantee}
Given an integer $k\ge 1$, and a graph $G$ that has maximum clique size $\omega$ and whose complement has maximum matching size~$\bar{\mu}$, 
we can solve $(\omega + \bar{\mu} - k)$-\Color{} in randomized $O^*(2^{6k})$ time.
\end{restatable}

\subparagraph*{Our approach.}

The key ideas behind \Cref{thm:dualcolor,thm:guarantee} are sketched below; see \Cref{subsec:technical-overview} for more details.
We first greedily compute a maximal packing of co-triangles $\overline{K_3}$ (independent triplets).
If the packing holds many co-triangles, saving two colors per co-triangle already yields the desired proper coloring.
Otherwise, the packing's vertices form a small \emph{co-triangle modulator} $S$, whose deletion results in a $\overline{K_3}$-free graph.
In the latter case, we extend the randomized procedure of Gutin et al.~\cite{GMOW2021}, to show that we can solve the desired coloring problem in $O^*(2^{|S|})$ time.
 
\subparagraph*{Hardness results.}
We also note that algorithms which are fixed-parameter tractable with respect to $k$ are unlikely to exist for $(\omega-k)$-\Color{} and $(\bar{\mu}-k)$-\Color{}, so that combining $\omega$ and $\bar{\mu}$ together in \Cref{thm:guarantee} appears to be necessary in order to achieve efficient parameterized algorithms with respect to these  parameters.
For the former, this is because 3-\Color{} is \NP{}-hard in planar graphs \cite{GJS1976}. 
Since a planar graph does not have a clique of size five, this means that assuming $\P\neq\NP$ we cannot hope to solve $(\omega-k)$-\Color{} in $f(k)\poly(n)$ time for any function $f(\cdot)$.
To show hardness for the latter problem, we prove the following.

\begin{restatable}{theorem}{guaranteehard}
\label{thm:hardness}
    The
    $(n/2 - k)$-\Color{} problem is {\normalfont\textsf{W[1]}}-hard in the parameter $k$ on graphs with $n$ vertices whose complement has a perfect matching.
\end{restatable}
\hspace{-1.5em}In view of \Cref{thm:hardness}, we cannot hope for an $O^*(f(k))$-time algorithm for $(\bar{\mu}-k)$-\Color{},
for any function $f(\cdot)$.

\subsection{Technical Overview}
\label{subsec:technical-overview}

As a starting point, we give a brief overview for how the 
work of \cite{GMOW2021} implies a
$O^*(4^k)$ time algorithm for \DualColor{}.

Given a host graph $G$ and pattern graph $H$, we say $G$ is \emph{$H$-free} if $G$ does not contain $H$ as an induced subgraph. 
We say a set of vertices $S$ in $G$ is an \emph{$H$-free modulator} if deleting the vertices in $S$ from $G$ makes the graph $H$-free. 
For $\ell \in \mathbb{N}$,
we let $K_\ell$ denote the complete graph on $\ell$ vertices,
and $\bar{K_\ell}$ denote its complement, the independent set on $\ell$ vertices.

The main relevant result from previous work is the following, 
proven in \cite[Theorem 3.5]{GMOW2021}.

\begin{proposition}[Clique Modulator]
    \label{prop:modulator}
    Given a graph $G$ together with a $\bar{K_2}$-free modulator of size $p$ (i.e., a set of at most $p$ vertices, whose deletion from $G$ turns the graph into a clique),
    we can solve $k$-\Color{} on $G$ for any $k$ in randomized $O^*(2^p)$ time.
\end{proposition}

Recall that given a pattern graph $H$ and a host graph $G$, we say an \emph{$H$-packing} in $G$  is a collection $\cal{H}$ of vertex-disjoint, induced copies of $H$ in $G$. 
Using \Cref{prop:modulator} we can solve \DualColor{} in $O^*(4^k)$ time using the following three steps.

\begin{description}
    \item[Step 1: Identify Matching]\hfill
    
    Extract a maximal $\overline{K_2}$-packing $M$ in $G$.
    \item[Step 2: Large Matching $\Rightarrow$ Easy Instance]\hfill
    
    If $|M|\ge k$, then report that a proper coloring using at most $(n-k)$ colors exists. 
    \item[Step 3: Small Matching $\Rightarrow$ Small Modulator]\hfill
    
    If $|M| < k$, the vertices participating in $M$ form a $\bar{K}_2$-free modulator of size $p < 2k$. 
    Then apply \Cref{prop:modulator} to solve
    \DualColor{} in 
    $O^*(4^k)$ time.
\end{description}

In \stepfont{Step 1} we obtain $M$ by running a polynomial time algorithm for finding a maximum matching in the complement graph $\bar{G}$.
In \stepfont{Step 2}
we observe that if $|M|\ge k$, then we can assign a distinct color for each $\bar{K_2}$ in $M$ to color its two vertices, and then color all remaining vertices in $G$ with different colors to obtain a proper coloring using at most \(k + (n-2k) = n-k\) colors overall.
In \stepfont{Step 3} we observe that if $|M| < k$,
by maximality deleting each $\bar{K_2}$ in $M$ results in a $\bar{K_2}$-free graph,
so \Cref{prop:modulator} solves the problem. 

This strategy of extracting a matching and then performing casework on its size to either identify a solution or enforce more structure is similar to previous work on kernelization for \DualColor{} \cite{ChorFJ04,LDYR2021} and applications of crown decomposition \cite[Chapter 4]{KernelizationTextbook}. 

To solve \DualColor{} faster in $O^*(2^{3k/2})$ time, we work with a structure more complicated than a matching: a triangle packing. 
The core technical lemma powering our new algorithms is the following result.

\begin{restatable}{lemma}{main}
\label{lemma:main}
Given a $\overline{K_3}$-free modulator of size $p$ and any positive integer $k$, the $k$-\Color{} problem can be solved in randomized $O^*(2^{p})$ time.
\end{restatable}

Since any $\bar{K_3}$-free modulator is also a $\bar{K_2}$-free modulator by definition, \Cref{lemma:main} is a stronger version of \Cref{prop:modulator}. 
We remark that an extension to $\overline{K_4}$-free modularator is unlikely because \Color{} is NP-hard on $\overline{K_4}$-free graphs \cite{DBLP:conf/wg/KralKTW01}.

\begin{description}
    \item[Step 1: Identify Triangle Packing]\hfill
    
        Extract a maximal $\bar{K_3}$-packing $\cal{T}$ in $G$.

    \item[Step 2: Large Packing $\Rightarrow$ Easy Instance]\hfill
    
        If $|\cal{T}|\ge k/2$, then report that a proper coloring using at most $(n-k)$ colors exists.

    \item[Step 3: Small Packing $\Rightarrow$ Small Modulator]\hfill

        If $|\cal{T}| < k/2$,
        the vertices participating in $\cal{T}$ form a $\bar{K_3}$-free modulator of size $p < 3k/2$.
        Then apply \Cref{lemma:main} to solve \DualColor{} in $O^*(2^{3k/2})$ time. 
\end{description}

In \stepfont{Step 1} we obtain $\cal{T}$ in polynomial time by running a greedy algorithm. 
In \stepfont{Step 2} we observe that if $|\cal{T}| \ge k/2$,
then we can assign each $\bar{K_3}$ in $\cal{T}$ a different color to color its three vertices,
and then color all remaining vertices in $G$ with different colors to obtain a proper coloring using at most 
    \(k/2 + (n - 3k/2) = n-k\)
colors overall.
In \stepfont{Step 3} we observe that if $|\cal{T}| < 3k/2$, by maximality deleting each $\bar{K_3}$ in $\cal{T}$ results in a $\bar{K_3}$-free graph, so \Cref{lemma:main} solves the problem.

The proof of \Cref{prop:modulator} from \cite{GMOW2021} is established using algebraic techniques.
The main idea is that given a $\bar{K_2}$-free modulator $S$ of size $p$, the induced subgraph on $V\setminus S$ is a clique,  so every vertex outside of $S$ must be assigned a different color. 
Solving the $k$-\Color{} problem can then be reduced to a problem similar to bipartite matching,
where one must determine for each vertex outside of $S$ the subset of vertices in $S$ it shares a color with (intuitively, the subset of $S$ it is ``matched to''). 
The authors use determinants to enumerate bipartite matchings in an auxiliary graph capturing this question, and apply algorithms for \emph{polynomial sieving} and \emph{fast subset convolution} to find the best matching (corresponding to an optimal coloring) and solve this problem in $O^*(2^p)$ time, comparable to the number of subsets of $S$ that must be considered. 

Our proof of \Cref{lemma:main} uses similar technical ingredients.
Here, because we are given a $\bar{K_3}$-free modulator $S$ instead of a $\bar{K_2}$-free modulator, we can no longer assume that the vertices in $\bar{S} = V\setminus S$ all get different colors.
However, we can assume that in a proper coloring of the graph, any given color appears at most twice across the nodes in $\bar{S}$.
This then lets us reduce the $k$-\Color{} problem in this context to a problem related to perfect matchings, albeit in a nonbipartite graph.
At a high level, we do this by considering all partitions of $\bar{S}$ into sets of size at most two, where these parts represent distinct color classes in $\bar{S}$, and then consider all ways of matching these parts to different subsets of $S$, corresponding to extending these color classes from $\bar{S}$ to the full graph. 

Since we work with matchings in general undirected graphs, we use Pfaffians instead of determinants in our arguments.
We also used a Pfaffian-based algebraic approach in recent work on \textsc{Edge Coloring} \cite{DBLP:conf/stacs/AkmalK25}.
On top of this framework, we use fast subset convolution to ensure that we achieve the same $O^*(2^p)$ runtime as before. 

Our algorithm for 
$(\omega + \bar{\mu}-k)$-\Color{} follows by employing a similar strategy to the procedure for \DualColor{} outlined above, with a more refined analysis with respect to the relevant structural parameters. 

\subparagraph*{Comparison with the $(n-k)$-\SetCover{} Result.}
Alferov et al.~\cite{AlferovBB24} achieved an $O^*(2^{3k/2}|\mathcal{F}|)$ algorithm for $(n-k)$-\SetCover{} on input families $\mathcal{F}$ via a different toolkit.  
Their procedure likewise begins by extracting a maximal triangle packing and immediately reports a yes-instance when that packing is large (\stepfont{Step 2} in our outline).  
When the packing is small, they study the matching structure in the remaining part via the Gallai-Edmonds decomposition.
A sufficiently large matching again certifies a solution, and they carefully use this fact to obtain a runtime bound of $O^*(2^{3k/2}|\mathcal{F}|)$.  
These arguments are tailored to $(n-k)$-\SetCover{}, and do not lead to an efficient algorithm parameterized by $\overline{K_3}$-free modulator size as in Lemma~\ref{lemma:main}, which is central to proving Theorem~\ref{thm:guarantee}.
For the special case of $(n-k)$-\Color{},
where the set $\mathcal{F}$ can be succinctly described as the collection of independent sets in a given graph, the dynamic program employed in \cite{AlferovBB24} still suffers from this issue.
Due to this limitation, this previous approach does not directly prove \Cref{thm:dualcolor} either.

\subparagraph*{Organization.}

In \Cref{sec:prelim}, we review notation and basic facts about graphs, matrices, and polynomials that will be useful to us. 
In \Cref{alg:modulator} we prove \Cref{lemma:main}.
In \Cref{sec:dualcoloring}, we apply
\Cref{lemma:main} to design our algorithm for \DualColor{} and $(\omega+\bar{\mu}-k)$-\Color{} (our harder structural parameter variant of \DualColor{}) and prove \Cref{thm:guarantee},
and also establish a lower bound for a related problem by proving \Cref{thm:hardness}.
We conclude in \Cref{sec:conclusion} by summarizing our work and mentioning relevant open problems.

\section{Preliminaries}
\label{sec:prelim}

\subparagraph*{General Notation.}

Given a positive integer $a$, we let $[a] = \set{1, \dots, a}$ denote the set of the first $a$ consecutive positive integers. 

\subparagraph*{Graph Notation.}
We let $G$ denote the input graph on $n$ vertices and $m$ edges. 
We let $V$ denote the vertex set of $G$. 
Given a subset of nodes $S\sub V$,
we let $G[S]$ denote the induced subgraph of $G$ on $S$. 
We let $\mu(G), \omega(G)$, and $\alpha(G)$ denote the size of a maximum matching, clique, and independent set in $G$ respectively. 
We let $K_3$ denote the complete graph on three nodes,
which we refer to as a \emph{triangle}.

\subparagraph*{From Coloring to Clique Covers.}

Given a graph $G$,
we say a \emph{clique cover} of $G$ is a collection $\cal{C}$
of vertex-disjoint cliques in $G$ such that every vertex appears in some clique of $\cal{C}$.
The \emph{size} of this cover is simply $|\cal{C}|$, the number of cliques used in the clique cover. 

In the $p$-\CliqueCover{} problem,
we are given a graph $G$ on $n$ vertices and are tasked with determining if $G$ has a clique cover of size at most $p$. 
For convenience, we refer to $(n-k)$-\CliqueCover{} as the \DualCliqueCover{} problem.
By complementing the graph it is easy to see that $p$-\Color{} and $p$-\Clique{} are equivalent computational tasks.

\begin{observation}[\CliqueCover{}]
    \label{obs:clique-cover}
    For any integer $p\ge 1$, 
    solving $p$-\Color{} on a graph $G$ is equivalent to solving
    $p$-\Clique{} on the complement graph $\bar{G}$.
\end{observation}
\begin{proof}
    Suppose we have a proper coloring of $G$.
    Then the set of vertices assigned any fixed color form an independent set in $G$, 
    which means they form a clique in $\bar{G}$.
    Thus the coloring induces a clique cover of the same size.
    The reverse direction, that a solution to $p$-\Clique{} on $\bar{G}$ implies a solution to $p$-\Color{} on $G$, follows by symmetric reasoning.
\end{proof}

For convenience, in our proofs we will use \Cref{obs:clique-cover}  to frame our algorithms as solving problems related to \CliqueCover{}  instead of \Color{}. 

\subparagraph*{Algebraic Preliminaries.}

Throughout we work over a finite field $\FF$ of size $\poly(n)$.
Arithmetic operations over $\FF$ take $\poly(\log n)$ time. 

Given a set $S$, let $\Pi(S)$ denote the set of perfect matchings on $S$ (i.e., the collection of partitions of $S$ into sets of size exactly two). 
Given a skew-symmetric matrix $\mA$ (i.e., a matrix that equals $\mA = -\mA^\top$ the negative of its transpose) with rows and columns indexed by a set $S$, 
its Pfaffian is defined to be 
    \begin{equation}
    \label{eq:Pfaffdef}
    \Pf \mA = \sum_{M\in \Pi(S)} \sgn(M)\prod_{\set{u,v}\in M} \mA[u,v]
    \end{equation}
for a function $\sgn \colon\Pi(S)\to\set{-1,1}$ whose definition is not relevant here \cite[Section 7.3.2]{murota1999matrices}.

An arithmetic circuit is a way of constructing a polynomial by starting with its input variables, and iteratively building up more complicated expressions by using the standard  operations of addition, multiplication, and division.
The size of an arithmetic circuit is the total number of non-scalar operations used to construct its final output polynomial in this way. 
Our algorithms use the well known fact that Pfaffians admit polynomial-size arithmetic circuits that only use addition and multiplication operations \cite{Rote2001}. 

\begin{proposition}[Pfaffian Construction]
    \label{prop:Pfaffian}
    In $\poly(n)$ time we can construct a division-free arithmetic circuit of $\poly(n)$ size over $\FF$ whose inputs are indeterminate entries of an $n\times n$ skew-symmetric matrix $\mA$, and whose output is the polynomial $\Pf \mA$.
\end{proposition}

We make use of the following classic result, proven for example in \cite[Theorem 7.2]{Motwani1995}.

\begin{proposition}[Identity Testing]
    \label{prop:PIT}
    Let $P$ be a nonzero polynomial over a finite field $\FF$ of degree at most $d$.
    If each variable of $P$ is assigned an independent, uniform random value from $\FF$,
    then the corresponding evaluation of $P$ is nonzero with probability at least $1-d/|\FF|$.
\end{proposition}

\subparagraph*{Subset Convolution.}

Let $S$ be a set of size $p$.
Let $\alpha$ and $\beta$ be functions from the collection of subsets of $S$ to $\FF$.

The \emph{subset convolution} $(\alpha\ast\beta)$ of $\alpha$ and $\beta$ is the function defined by setting
    \[(\alpha\ast\beta)(T) = \sum_{\substack{A, B\sub S\\ T = A\sqcup B}} \alpha(A)\beta(B)
\quad\text{for all $T\sub S$}. \]

For each $v\in S$, introduce a variable $y_v$. Let $Y = \set{y_v}_{v\in S}$ be this set of  variables. 
Consider polynomials
    \[P = \sum_{A\sub S} \alpha(A)\prod_{a\in A} y_a
\quad\text{ and }\quad
    Q = \sum_{B\sub S} \beta(B)\prod_{b\in B} y_b\]
over the quotient ring $R = \FF[Y]/\langle (y^2)_{y\in Y}\rangle$.
That is, $R$ is the usual ring of polynomials over the variables in $Y$  with coefficients from $\FF$, except that whenever we multiply the same variable with itself the resulting product vanishes. 

We refer to $R$ as the \emph{squarefree ring} over $Y$. 
By definition of polynomial multiplication and the fact that only squarefree monomials survive in $R$,
we get that 
    \[P\cdot Q = \sum_{T\sub S} (\alpha\ast\beta)(T)\prod_{t\in T} y_t\]
over $R$. 
This shows that computing subset convolutions is equivalent to computing products of polynomials over $R$. 
The following result then follows from known $O^*(2^p)$ time algorithms for subset convolution \cite{subset-convolution}.

\begin{proposition}[Fast Subset Convolution]
\label{prop:convolution}
    We can compute addition and multiplication over the squarefree ring $R$ on $p$ variables in $O^*(2^p)$ time.
\end{proposition}

\section{Triangle Modulator}
\label{alg:modulator}

In this section, we prove the following key lemma.

\main*
\begin{proof}
    Replace the input graph with its complement.
    Then by \Cref{obs:clique-cover}, it suffices to show that 
    given a $K_3$-free modulator $S$ of size $p$ for $G$ and a positive integer $k$,
    we can solve $k$-\Clique{} on $G$ in $O^*(2^p)$ time. 

    We write $\bar{S} = V\setminus S$.
    By definition, $G[\bar{S}]$ contains no triangle (i.e., a copy of $K_3$). 
    This means that for every clique $C$ in $G$,
    we have $|C\cap \bar{S}|\in\set{0,1,2}$.
    Given a clique cover $\cal{C}$ of $G$,
    we say it has type $\vec{t} = (t_0, t_1)$ if for each $i\in\set{0,1}$ the cover $\cal{C}$ contains exactly $t_i$ cliques $C$ satisfying $|C\cap\bar{S}| = i$.
    Note that if a cover has type $\vec{t}$,
    then 
    because every vertex appears in a unique clique in the cover,
    exactly $t_2 = (|\bar{S}| - t_1)/2$ cliques in the cover intersect $\bar{S}$ at two nodes.
    Consequently, the size of a clique cover with type $\vec{t}$ is exactly
        \[t_0 + t_1 + t_2 = t_0 + (t_1 + |\bar{S}|)/2 = t_0 + (t_1 + n-p)/2.\]
    We say a type $\vec{t} = (t_0, t_1)$ is \emph{valid} if and only if we have 
        \(t_0 + (t_1 + n-p)/2 \le k\)
    so that a clique cover of type $\vec{t}$ has size at most $k$. 

    Our task is to determine if $G$ contains a clique cover of size at most $k$.
    Such a cover has at most $O(k^2) \le \poly(n)$ valid types $\vec{t}$, and we try out each possible such type and look for a clique cover with exactly that type. 

    To that end, fix a valid type $\vec{t} = (t_0, t_1)$.
    We now construct an auxiliary graph $H$,
    such that perfect matchings in $H$  encode how clique covers of $G$ with type $\vec{t}$ may restrict to clique covers of $G[\bar{S}]$.
    We include the vertex set $\bar{S}$ in $H$, and add all edges from $G[\bar{S}]$ to $H$.
    We then introduce a set $U$ of $t_1$ new vertices, and add edges from every node in $U$ to every vertex in $\bar{S}$.
    Intuitively, when we take a perfect matching $M$ of $H$, it will correspond to a clique cover $\cal{C}$ of $G$ such that 
    \begin{itemize}
        \item for each edge $\{ u, v \}$ with $u\in U$, there is a clique $C\in\cal{C}$ with $C\cap\bar{S} = \set{v}$, and
        \item for each edge $\{ v, w \}\in M$ with $v,w\in \bar{S}$, there is a clique $C\in \cal{C}$ with $C\cap\bar{S} = \set{v,w}$.
    \end{itemize}

    With this correspondence in mind, we 
    move to constructing a polynomial that will enumerate clique covers of $G$.

    Let $E(H)$ denote the edge set of $H$.
    For each $e\in E(H)$, introduce an indeterminate variable $x_e$.
    Now consider the matrix $\mA$ whose rows and columns are indexed by nodes of $H$,
    with entries defined by 
        \[\mA[v,w] = \begin{cases}
            x_{vw}&\text{if }\set{v,w}\in E(H)\text{ and } v\prec w \\
            -x_{vw}&\text{if }\set{v,w}\in E(H)\text{ and } w\prec v \\
            0&\text{otherwise}
        \end{cases}\]
    where $(\prec)$ is some fixed ordering on the nodes of $H$. 

    By \Cref{eq:Pfaffdef} and the definition of $\mA$ above, we have 
        \begin{equation}
        \label{eq:Pf-A}
        \Pf\mA = \sum_{M\in\Pi(H)} \sgn(M)\prod_{e\in M} x_e
        \end{equation}
    where here $\Pi(H)$ denotes the set of perfect matchings in $H$,
    and $\sgn(M)\in\set{-1,1}$ for each choice of $M$. 

    Next, we will introduce additional, larger polynomial expressions that we will substitute in for the $x_e$ variables,
    with the end goal of transforming the matching polynomial $\Pf \mA$ into a
    polynomial that enumerates collections of cliques in $G$.
    
    Go through all subsets $T\sub S$,
    and for each check if $T$ is a clique in $G$.
    This lets us determine the collection $\cal{F}$ of all cliques in $S$ in $O^*(2^p)$ time overall.

    Then for each edge $e\in E(H)$ and clique $C\sub S$, we introduce a variable $z_{eC}$. 
    For each $i\in [t_0]$ and clique $C\sub S$, also introduce a variable $z_{iC}$.
    Let $Z$ be the set of all these $z_{eC}$ and $z_{iC}$ variables.
    
    Additionally, for each each node $v\in S$, we introduce a variable $y_v$.
    Let $Y = \set{y_v}_{v\in S}$ denote this set of $k$ variables.

    Now for each $e\in E(H)$ and subset $T\sub S$,
    define the polynomial
        \begin{equation}
        \label{eq:Phi-poly}
        \Phi_{e}[T] = \sum_{\substack{C\in \cal{F} \\ C\sub T}} z_{eC}\prod_{v\in C} y_v.
        \end{equation}
    Intuitively this polynomial enumerates those cliques $C\sub T$ that can be extended 
    using vertices in $e\cap\bar{S}$ to obtain larger cliques in $G$.
    
        Also for each $i\in [t_0]$, define the polynomial
        \begin{equation}
        \label{eq:Phi-i}
        \Phi_i = \sum_{C\in\cal{F}} z_{iC}\prod_{v\in C} y_v.
        \end{equation}
    Intuitively these polynomials enumerate cliques $C\sub S$ that we do not plan on extending with nodes in $\bar{S}$.

    For each vertex $v\in V$,
    let $N(v)$ denote the set of vertices in $S$ adjacent to $v$ in $G$.
    We construct a new matrix $\mB$,
    by starting with the matrix $\mA$ and 
    \begin{itemize}
        \item for each $e = \set{u,v} \in E(H)$ with $u\in U$ and $v\in \bar{S}$, substituting $x_{uv}$ with 
            $\Phi_{e}[N(v)]$, and
        \item 
            for each $e = \set{v,w}\in E(H)$ with $v,w\in \bar{S}$, substituting $x_{vw}$ with $\Phi_e[N(v)\cap N(w)]$.
    \end{itemize}

    Now define the polynomial
        \begin{equation}
        \label{eq:Fdef}
        F = (\Pf \mB)\cdot \prod_{i=1}^{t_0} \Phi_i
        \end{equation}
    over $R[Z]$, where   $R = \FF[Y]/\pair{(y^2)_{y\in Y}}$ is the squarefree ring defined in \Cref{sec:prelim}.

    \begin{claim}[Polynomial Characterization]
    \label{claim:poly-char}
        The polynomial $F$ has a monomial divisible by \(\prod_{v\in S} y_v\)
        if and only if $G$ has a clique cover of type $\vec{t} = (t_0, t_1)$.
    \end{claim}
    \begin{claimproof}
        Suppose first that $F$ has a monomial divisible by $\prod_{v\in S} y_v$.
        Since we work over the squarefree ring $R$, this monomial is of the form
            \begin{equation}
            \label{eq:mon}
            f(Z)\prod_{v\in S} y_v
            \end{equation}
        for some polynomial $f\in \FF[Z]$.
        Moreover, from \Cref{eq:Fdef} such a monomial in $F$ must have been generated by taking the product of monomials from $\Pf\mB$ and each of the $\Phi_i$ such that the appearance of variables from $Y$ in these monomials partition $Y$. 
        
        For each $i\in [t_0]$, let $C_i$ be the set of vertices $v\in S$ such that the $y_v$ variable was used by the monomial from $\Phi_i$ selected to help generate the expression from \Cref{eq:mon} in the expansion of the product from \Cref{eq:Fdef}.
        Then by the partition property, the $C_i$ are pairwise disjoint.
        From the definition of $\Phi_i$ in \Cref{eq:Phi-i}, we get that each $C_i$ is a clique.

        Now, consider the monomial selected from $\Pf \mB$ to help generate the expression from \Cref{eq:mon} in the expansion of the product from \Cref{eq:Fdef}.
        By \Cref{eq:Pf-A} and the definition of $\mB$ in terms of $\mA$,
        this monomial corresponds to a matching $M \in \Pi(H)$.
        Split $M = M_1\sqcup M_2$ by letting $M_b$ retain the edges in $M$ 
        which have exactly $b$ endpoints in $\bar{S}$ for each $b\in [2]$.
        Since $H$ has vertex set $\bar{S}\sqcup U$ and $|U| = t_1$, 
        we get that $M_1$ has exactly $t_1$ edges, and each of its edges has one endpoint in $U$ and one endpoint in $\bar{S}$.
        Let $v_1, \dots, v_{t_1}$ be the nodes in $\bar{S}$ appearing as endpoints of edges in $M_1$.
        For each $j\in [t_1]$, let $e_j$ be the edge containing $v_j$ in $M_1$.

        Since in $\mB$ we take $\mA$ and replace each $x_{e_j}$ variable
        with $\Phi_{e_j}[N(v_j)]$, the 
        product over the edges of $M_1$ from \Cref{eq:Pf-A}
        contributes the variables $y_v$ precisely for those vertices $v$ appearing in the union of some choice of cliques $D_1, \dots, D_{t_1}\sub S$ with the property that the set $D_j\sqcup\set{v_j}$ is a clique in $G$ for each $j$. 
        Since we work over $R$, the cliques $C_i$ and $D_j$ are  mutually disjoint.
        Moreover, since $C_i\sub \bar{S}$ for all $i$,
        we actually have that the cliques $C_i$ and $D_j\sqcup\set{v_j}$ are collectively vertex-disjoint.

        Set $t_2 = (|\bar{S}| - t_1)/2$, and 
        let $\tilde{e}_1, \dots, \tilde{e}_{t_2}$ be the edges of $M_2$.
        For each $\ell\in [t_2]$,
        let $\tilde{N}_{\ell}$ be the set of common neighbors in $S$ of the endpoints of the edge $\tilde{e}_\ell$.
        Since in $\mB$ we take $\mA$ and replace each $x_{\tilde{e}_\ell}$ variable
        with $\Phi_{\tilde{e}_\ell}[\tilde{N}_\ell]$,
        the product over the edges of $M_2$ from \Cref{eq:Pf-A}
        contributes the variables $y_v$ precisely for those vertices $v$ appearing
        in the union of some choice of cliques $\tilde{D}_1, \dots, \tilde{D}_{t_2}\sub S$
        with the property that the set $\tilde{D}_\ell\sqcup e_\ell$ is a clique in $G$ for each $\ell$.
        Since we work over $R$, the $C_i$, $D_j$, and $\tilde{D}_\ell$ are mutually disjoint.
        Moreover, since $M=M_1\sqcup M_2$ is a matching in $S$, 
        we in fact know that the $C_i$, $D_j\sqcup\set{v_j}$, and $\tilde{D}_\ell\sqcup \tilde{e}_\ell$ are all vertex-disjoint cliques in $G$.

        Combining this vertex-disjoint condition with the facts that $M$ is a perfect matching on $\bar{S}$ and that the monomial in \Cref{eq:mon} generated by the products of the monomials of all the cliques we have listed so far is divisible by $\prod_{v\in S} y_v$,
        we get that the 
        \(C_i\text{, }D_j\sqcup\set{v_j}\text{, and }\tilde{D}_{\ell}\sqcup\tilde{e}_\ell\) cliques taken together form a clique cover of $G$ of type $\vec{t} = (t_0, t_1)$ as claimed.

        Conversely, suppose $G$ has a clique cover of type $\vec{t} = (t_0, t_1)$.
        Set $t_2 = (|\bar{S}| - t_1)/2$.
        
        Let $C_1, \dots, C_{t_0}$ be the cliques in this cover intersecting $\bar{S}$ at zero nodes, 
        $D_1, \dots, D_{t_1}$ be the cliques in this cover intersecting $\bar{S}$ at one node,
        and $\tilde{D}_1, \dots, \tilde{D}_{t_2}$ be the cliques in this cover intersecting $\bar{S}$ at two nodes. 
        For each $i\in [t_0]$, select the monomial from $\Phi_i$ in \Cref{eq:Phi-i} corresponding to the clique $C_i$.
        
        For each $j\in [t_1]$,
        let $v_j$ be the unique vertex from $\bar{S}$ in $D_j$. 
        For each $\ell\in [t_2]$, let $\tilde{e}_\ell$ be the subset of two vertices in $\bar{S}$ contained in $\tilde{D}_\ell$.
        Since $\tilde{D}_\ell$ is a clique, $\tilde{e}_\ell$ is an edge. 
        Let $\tilde{N}_\ell$ be the set of common neighbors in $\bar{S}$ of the endpoints of $\tilde{e}_\ell$.
        The $v_j$ and $\tilde{e}_\ell$ are vertex-disjoint and account for all vertices in $\bar{S}$, because we are assuming we are starting with a clique cover.
        So we can take a perfect matching $M$ of $H$ which includes all the $\tilde{e}_\ell$ edges, and pairs each $v_j$ with a different node in $U$ (here we use the fact that $U$ has exactly $t_1$ nodes). 
        Let $e_j$ be the edge containing $v_j$ in $M$ for each $j\in [t_1]$.
        
        Then we can select the monomial of $\Pf \mB$ corresponding to choosing the summand for $M$ in \Cref{eq:Pf-A}, choosing the summand for the cliques $D_j\setminus\set{v_j}$ in the sums for $\Phi_{e_j}[N(v_j)]$ given by \Cref{eq:Phi-poly}, and choosing the summands for the cliques $\tilde{D}_\ell\setminus\tilde{e}_\ell$ in the sums for $\Phi_{\tilde{e}_\ell}[\tilde{N}_\ell]$ given by \Cref{eq:Phi-poly} again. 

        Then from the clique cover assumption, the product of all the monomials we selected from the $\Phi_i$ and $\Pf \mB$ will be divisible by $\prod_{v\in S} y_v$.
        Moreover, the product of the variables from $Z$ appearing in this monomial uniquely recover the cliques $C_i$, $D_j$, and $\tilde{D}_\ell$, as well as the matching $M$, because the $z_{iC}$ variables record the ordering of the $C_i$ cliques, and the $z_{eC}$ variables annotate the edges $e$ of the matchings corresponding to the $D_j\cap\bar{S}$ and $\tilde{D}_\ell\cap\bar{S}$ sets.
        Thus the monomial generated in this way cannot be cancelled out by any other term in the expansion of $F$ in \Cref{eq:Fdef}, which proves the desired result. 
        \end{claimproof}

Having established \Cref{claim:poly-char},
it suffices to show that we can test that $F$ has a monomial divisible by  $\prod_{v\in S} y_v$ in $O^*(2^p)$ time.
To do this, we will take a random evaluation of the variables in $Z$, and then use fast subset convolution over the variables in $Y$.

For all $i\in [t_0]$ and cliques $C\sub S$, we pick independent, uniform random $\xi_{iC}\in\FF$.
Independently from those values, for all $e\in E(H)$ and $C\sub S$ we also pick independent, uniform random $\xi_{eC}\in\FF$.
For each $i\in [t_0]$, let 
    \begin{equation}
        \label{eq:Phi-i-random-eval}
        \varphi_i = \sum_{C\in\cal{F}} \xi_{iC}\prod_{v\in C} y_v.
    \end{equation}
    be the result of evaluating the $Z$ variables of $\Phi_i$ on the $\xi_{iC}$ values. 
    We can compute each $\varphi_i$ in $O^*(2^p)$ time because we have already precomputed the collection $\cal{F}$ of cliques in $S$.
    
    Similarly,
    for each $e\in E(H)$ and $T\sub S$, let 
        \begin{equation}
        \label{eq:Phi-poly-random-eval}
        \varphi_e[T] = \sum_{\substack{C\in \cal{F} \\ C\sub T}} \xi_{eC}\prod_{v\in C} y_v
    \end{equation}
    be the result of evaluating the $Z$ variables of $\Phi_{e}[T]$ on the $\xi_{eC}$ values.
  Since we precomputed all the cliques contained in $S$,
  for any fixed $e\in E(H)$ and $T\sub S$ we can use \Cref{eq:Phi-poly-random-eval}
  to compute $\varphi_e[T]$ in $O^*(2^p)$ time. 
In particular, we can compute $\varphi_e[T]$
        for all $e\in E(H)$ and $T\sub S$ of the form 
        $T=N(v)$ for some $v\in\bar{S}$ or $T=N(v)\cap N(w)$ for some $v,w\in \bar{S}$
        in $O^*(2^p)$ time,
        because there are $\poly(n)$ choices for $e$ and $T$ in this case. 

    Having computed all these values, we can substitute the $x_e$ for the relevant values $\varphi_e[T]$ values needed to turn $\mA$ into $\mB$ (with respect to our random evaluation of the variables in $Z$).
    Now use \Cref{prop:Pfaffian} to obtain a polynomial-size, division-free arithmetic circuit for the Pfaffian, and feed the entries of $\mB$ we have just computed into it.
    Each addition and multiplication operation for this arithmetic circuit is now over the squarefree ring $R$, which by \Cref{prop:convolution} takes $O^*(2^p)$ time.
    Since there are only $\poly(n)$ such operations,
    we compute this random evaluation of $\Pf \mB$ over $R$ in $O^*(2^p)$ time overall.
   We then compute the product of this with all the $\varphi_i$ again in $O^*(2^p)$ time by \Cref{prop:convolution},
   which by \Cref{eq:Fdef} gives us the value of the polynomial $F$ under our random evaluation to its $Z$ variables.

   Since $F$ has degree at most $\poly(n)$ in its $Z$ variables,
   applying \Cref{prop:PIT} to the coefficient of $\prod_{v\in S}y_v$ in $F$ (viewed as a polynomial in $Z$), we see that by picking the size of the field $\FF$ to be a sufficiently large polynomial in $n$, with high probability the random evaluation of $F$ we computed has the monomial $\prod_{v\in S}y_v$ if and only if the original polynomial $F$ has a monomial divisible by $\prod_{v\in S} y_v$. 
   So by \Cref{claim:poly-char} checking if this random evaluation of $F$ has the monomial $\prod_{v\in S}y_v$ solves the $k$-\Color{} problem in $O^*(2^p)$ time with high probability.
   This proves the desired result.
   \end{proof}

\section{Below Guarantee Parameterizations}
\label{sec:dualcoloring}

In this section, we present our algorithm for \DualColor{}. 

\dualcolor*
\begin{proof}
    We replace the graph $G$ with its complement.
    Then by \Cref{obs:clique-cover}, it suffices to solve the \DualCliqueCover{} problem.

    We first construct a maximal collection $\cal{T}$ of vertex-disjoint triangles in $G$. 
    Since we can find a triangle in $G$ if one exists in polynomial time, we can construct $\cal{T}$ by repeatedly finding a triangle $T$ in $G$, including $T$ in $\cal{T}$, then deleting the vertices of $T$ from $G$ and repeating this process, until we determine that no triangles are left.
    We find at most $n/3$ triangles in this way, so this takes polynomial time overall.

    Let $t = |\cal{T}|$.
    Suppose $t\ge k/2$.
    Then taking the triangles in $\cal{T}$ together with the single-node sets consisting of each vertex $v$ not used by a triangle in $\cal{T}$ yields a clique cover for $G$ of size  
        \(t + (n-3t) = n-2t \le n-k.\)
     Thus, we can solve \DualCliqueCover{} by returning the cover described above.

     Otherwise, we have $t < k/2$.
     In this case, maximality of $\cal{T}$ implies that the set 
        \(S = \bigsqcup_{T\in\cal{T}} T\)
    is a $K_3$-free modulator for $G$ of size $3t$.
    Then by \Cref{obs:clique-cover} and \Cref{lemma:main} we can solve \DualColor{} in this case in $O^*(2^{3t}) \le O^*(2^{3k/2})$ time, as desired. 
\end{proof}

Next, we establish (randomized) fixed-parameter tractability for a smaller parameter.

\guaranteefpt* 

\begin{proof}

    We replace the graph $G$ with its complement. 
    Let $\mu = \mu(G)$ and $\alpha = \alpha(G)$ be the size of a maximum matching and maximum independent set in the new graph respectively.
    Since complementing the graph turns cliques into independent sets,
by \Cref{obs:clique-cover}
    it suffices to show that we can solve $(\alpha+\mu-k)$-\CliqueCover{} in $O^*(2^{6k})$ time. 
    
    We first construct a maximal collection $\cal{T}$ of vertex-disjoint triangles in $G$.
    Since finding a triangle in $G$, if it exists, takes polynomial time, and $\cal{T}$ can have at most $n/3$ triangles, we can obtain $\cal{T}$ in polynomial time using a greedy algorithm that repeatedly finds triangles, includes them in $\cal{T}$, deletes the obtained triangles from $G$, and continues in this fashion until we are left with a triangle-free graph.  

    Let $t = |\cal{T}|$.
    Suppose that $t\le 2k$.
    In this case, define the set 
        \(
        S = \bigsqcup_{T\in\cal{T}} T
        \)
    of all vertices participating in the triangles of $\cal{T}$.
    Then by maximality of $\cal{T}$, the set $S$ is a $K_3$-free modulator of $G$ of size $3t$.
    In this case, by \Cref{obs:clique-cover} and \Cref{lemma:main} we can solve 
    $(\alpha+\mu-k)$-\CliqueCover{} in $O^*(2^{3t}) \le O^*(2^{6k})$ time as claimed.

    Otherwise, $t > 2k$.
In this case,
take $\tilde{\cal{T}}\sub \cal{T}$
consisting of exactly $2k$ triangles, and define
        \(
        S = \bigsqcup_{T\in\tilde{\cal{T}}} T
        \)
   to  be a set of $6k$ vertices making up $2k$ of the triangles  in $\cal{T}$.
We show that in this case, a clique cover of size at most $(\alpha+\mu-k)$ always exists in $G$.

Write $\bar{S} = V\setminus S$.
Let $M_{\t{out}}$ be a maximum matching in $G[\bar{S}]$.
We can find $M_{\t{out}}$ in polynomial time. 
Let $\mu_{\t{out}} = |M_{\t{out}}|$ be the size of this matching.
Let $I$ be the set of vertices in $\bar{S}$ that do not appear as an endpoint of an edge in $M_{\t{out}}$.
By maximality of $M_{\t{out}}$, the set $I$ forms an independent set in $G$. 
Write $\alpha_{\t{out}} = |I|$.

\begin{claim}
    \label{claim:ineq}
    We have $\mu_{\t{out}} + \alpha_{\t{out}} \le \alpha + \mu - 3k$.
\end{claim}
\begin{claimproof}
We will prove the claim by establishing several inequalities concerning structural parameters in $G$.

    First, construct a maximum matching $M_{\t{in}}$ in $G[S\sqcup I]$ in polynomial time.
    Let $\mu_{\t{in}} = |M_{\t{in}}|$ be the size of this matching. 
    Since $M_{\t{out}}$ is in $G[\bar{S}\setminus I]$ by definition,
    we see that $M_{\t{in}}\sqcup M_{\t{out}}$ is a matching in $G$ of size $(\mu_{\t{in}} + \mu_{\t{out}})$.
    
    Since $\mu$ is the size of a maximum matching in $G$,
    we get that 
        \begin{equation}
        \label{ineq:2}
        \mu_{\t{in}} + \mu_{\t{out}} \le \mu.
        \end{equation}

    Note that since $\alpha$ is the maximum size of independent sets in $G$, we have
    \begin{equation}
        \label{ineq:1}
        \alpha_{\t{out}} \le \alpha.
    \end{equation}

    Now, by maximality of $M_{\t{in}}$ in $G[S\sqcup I]$, we know that the set of vertices in $S\sqcup I$ not participating as an endpoint in $M_{\t{in}}$ forms an independent set in $G$.
    This set then has size $(\alpha_{\t{out}} + 6k - 2\mu_{\t{in}})$.
    Since $\alpha$ is the maximum size of independent sets in $G$, we have 
        \begin{equation}
        \label{ineq:3}
        \alpha_{\t{out}} + 6k - 2\mu_{\t{in}} < \alpha.
        \end{equation}

    Now, we can write 
        \begin{equation}
        \label{eq:average}
        \mu_{\t{out}} + \alpha_{\t{out}} = \mu_{\t{out}} + (1/2)\alpha_{\t{out}} + (1/2)\alpha_{\t{out}}.
        \end{equation}

    By applying \Cref{ineq:1,ineq:2,ineq:3} to the first through third terms on the right-hand side of \Cref{eq:average} respectively, we can bound
        \[\mu_{\t{out}} + \alpha_{\t{out}} \le (\mu - \mu_{\t{in}}) + (1/2)\alpha + (1/2)(\alpha + 2\mu_{\t{in}} - 6k) =  \alpha + \mu - 3k \]
    which proves the desired result. 
\end{claimproof}

By definition, every vertex in $\bar{S}$ either appears as an endpoint of $M_{\t{out}}$,
or belongs to $I$.
Then using the triangles in $\tilde{\cal{T}}$ to cover the vertices in $S$,
and taking the edges from $M_{\t{out}}$ and the individual nodes from $I$ to cover the vertices in $\bar{S}$,
we obtain a clique cover for $G$ of size 
    \[2k + \mu_{\t{out}} + \alpha_{\t{out}} \le \alpha + \mu - k\]
by \Cref{claim:ineq}.
So we can solve $(\alpha+\mu-k)$-\CliqueCover{} just by returning this cover. 
\end{proof}

Curiously, previous work showed that \textsc{Induced Matching} (the problem of finding a 1-regular induced subgraph on $k$ vertices) is also fixed-parameter tractable in $k$ for the same below-guarantee parameter $\alpha + \mu - k$ as in the proof above \cite{Koa23}.

\subparagraph*{Hardness.}

We now present our lower bound,
demonstrating that although \Cref{thm:guarantee}
shows we can get an efficient parameterized algorithm for graph coloring below $(\omega+\bar{\mu})$,
we do not expect an analogous algorithm to exist for coloring below the quantity $\bar{\mu}$ on its own.

\guaranteehard*
\begin{proof}
By replacing the graph with its complement and applying \Cref{obs:clique-cover},
it suffices to show that $(n/2 - k)$-\CliqueCover{} is \textsf{W[1]}-hard on graphs containing a perfect matching. 

We prove this result by reducing from $k$-\textsc{Colored Clique}.
In this problem, we are given a $k$-partite graph $G$ on $kn$ vertices,
with vertex set $V = V_1\sqcup \dots \sqcup V_k$ partitioned into parts $V_i$ consisting of $n$ nodes each,
and are tasked with determining if $G$ contains 
a clique on $k$ vertices.
This problem is known to be \textsf{W[1]}-hard \cite[Theorem 13.25]{ParameterizedTextbook}.

Take an instance $G$ of $k$-\textsc{Colored Clique}.

For each $i\in [k]$, we order the vertices in $V_i = \set{v_{i1}, \dots, v_{in}}$.

We construct a larger graph $\tilde{G}$
that contains $G$ as a subgraph,
in addition to some auxiliary nodes and edges we describe next. 
For each $i\in [k+2]$, we introduce a new node $u_i$ in $\tilde{G}$.
For each $i\in [k]$ and $j\in [n-1]$, we introduce a new node $w_{ij}$ in $\tilde{G}$.
We add edges between all the $u_i$ nodes. 
For all $i\in [k]$, we add an edge from $u_i$ to $v_{i1}$.
For all $i\in [k]$ and $j\in [n-1]$ we add edges from $w_{ij}$ to $v_{ij}$ and $v_{i(j+1)}$.

Conceptually, this construction of $\tilde{G}$ from $G$ involves adding a new clique of size $(k+2)$ to the graph on the $u_i$ vertices, and for each $i\in [k]$ introducing a path 
\begin{equation}
    P_i = \pair{u_i, v_{i1}, w_{i1}, v_{i2}, w_{i2}, \dots, w_{i(n-1)}, v_{in}}
\end{equation}
beginning at $u_i$, and alternating between  nodes in $V_i$ and $W_i$ according to the orders of the vertices in these sets. 
By design, for any $j\in [n]$, we can always split the path 
    \begin{equation}
        \label{eq:path-split}
        P_i = \pair{u_i}\diamond A_{ij} \diamond \pair{v_{ij}} \diamond B_{ij}
    \end{equation}
into its first node, a path $A_{ij}$ on an even number of vertices, the node $v_{ij}$, and a suffix $B_{ij}$ on an even number of vertices.
Moreover, $A_{ij}$ and $B_{ij}$ consist of $2n-2$ nodes altogether in this decomposition. 

Note that the constructed graph $\tilde{G}$ has a perfect matching, by taking the edges $\set{u_i, v_{i1}}$ for all $i\in [k]$ and $\set{w_{ij}, v_{i(j+1)}}$ for all $i\in [k]$ and $j\in [n-1]$.

From its definition, $\tilde{G}$ has $N = 2kn + 2$ nodes.
Set the parameter $\ell = k(n-1) + 2$.

\begin{claim}
\label{claim:reduction}
$G$ has a clique on $k$ vertices if and only if $\tilde{G}$ has a clique cover of size $\ell$.
\end{claim}
\begin{claimproof}
Suppose that $G$ has a clique of size $k$.
Let $(v_{1j_1}, \dots, v_{nj_n})\in V_1\times \dots \times V_k$ be the $k$-tuple of vertices participating in this clique.

For each $i\in [k]$, by \Cref{eq:path-split} we can split each path $P_i$ into 
    \[\pair{u_i}\diamond A_{ij_i}\diamond \pair{v_{ij_1}}\diamond B_{ij_1}\]
its first node, two paths $A_{ij_i}$ and $B_{ij_1}$ on even numbers of vertices with $2n-2$ nodes total, and the node $v_{ij_1}$.
For each $i\in [k]$, we can cover the nodes in $A_{ij_i}$ and $B_{ij_i}$ using the endpoints of the edges  from a matching of size $(n-1)$.
This shows that with $k(n-1)$ cliques, we can cover all vertices in $\tilde{G}$ in a disjoint fashion, except the $u_i$ and $v_{ij_i}$ nodes.
We include all the $u_i$ nodes in one clique and all the $v_{ij_i}$ nodes in another clique to then obtain a clique cover of size $\ell = k(n-1) + 2$ of $\tilde{G}$ as claimed.

Conversely, suppose $\tilde{G}$ has a clique cover of size $\ell$.
Note that the $w_{ij}$ vertices, with indices ranging over $i\in [k]$ and $j\in [n-1]$, form an independent set.
Consequently, covering the $w_{ij}$ requires using $k(n-1)$ distinct cliques.
Now, each $w_{ij}$ has degree two in $\tilde{G}$, and the neighbors of $w_{ij}$ are non-adjacent in $\tilde{G}$ by definition.
Consequently, the cliques used to cover the $w_{ij}$ collectively cover at most $2k(n-1)$ vertices in $\tilde{G}$.

In the assumed clique cover of size $\ell$ from $\tilde{G}$,
this implies that a set of $\ell - k(n-1) = 2$ distinct cliques are used to cover the at least $N - 2k(n-1) = 2k+2$ remaining nodes in $\tilde{G}$ not covered by the cliques that contain the $w_{ij}$ vertices.
Since node $u_{k+2}$ is adjacent to only other $u_i$ nodes, the clique that covers $u_{k+2}$ must solely consist of $u_i$ nodes.
Without loss of generality, we may then assume the clique covering $u_{k+2}$ contains all $u_i$ nodes.
After removing the nodes covered by this clique, we have one clique left, that must be used to cover the at least $(2k+2) - (k+2) = k$ remaining nodes in $\tilde{G}$.
Moreover, this clique uses no $u_i$ nodes or $w_{ij}$ nodes.
Consequently, the last clique must be a clique of size $k$ in the $v_{ij}$ nodes, which pulls back to the a clique of size $k$ in $G$.
This proves the claim.
\end{claimproof}

Observe that $N/2 - \ell = (k+1)$,
so that $\ell = N/2 - (k+1)$.

By \Cref{claim:reduction}, solving  $k$-\textsc{Colored Clique} problem on graphs with $n$ vertices reduces to solving $(N/2 - (k+1))$-\CliqueCover{} on graphs with $N = 2kn+2$ nodes and a perfect matching.
Hence the reduction is efficient enough to imply that $(n/2-k)$-\CliqueCover{} is \textsf{W[1]}-hard in graphs with perfect matchings, which proves the desired result.
\end{proof}
 
\section{Conclusion}
\label{sec:conclusion}

In this paper, we presented a faster parameterized algorithm for \DualColor{},
improving the runtime from $O^*(4^k)$ to $O^*(2^{3k/2}) \le O^*(2.83^k)$. 
We also introduced a new below-guarantee parameterization for graph coloring,
which can be viewed as a harder version of \DualColor{}.
We showed that this problem can be solved in $2^{O(k)}\poly(n)$ time as well,
and noted that a closely related parameterization is in contrast \textsf{W[1]}-hard.

The most relevant open question to this work is: what is the true parameterized complexity of \DualColor{}?
The current best algorithms for $k$-\Color{} take $O^*(2^n)$ time for general $k$, and without improving this runtime we cannot hope to solve \DualColor{} in faster than $O^*(2^k)$ time. Can we indeed achieve a $O^*(2^k)$ runtime for \DualColor{}? 

A natural strategy to obtain faster algorithms for \DualColor{} would be to try and strengthen \Cref{lemma:main} further, by obtaining efficient coloring algorithms parameterized by $H$-free modulators for larger pattern graphs $H$. 
For example, obtaining an analogue of \Cref{lemma:main} for $\bar{K_4}$-free modulators would immediately imply a faster \DualColor{} algorithm by following the framework outlined in \Cref{subsec:technical-overview}.
Unfortunately, this simple attempt at generalizing \Cref{lemma:main} does not seem possible.
This is because solving $k$-\Color{} in graphs without induced copies of $\bar{K}_4$ is $\NP$-hard{} in general \cite[Theorem 1]{DBLP:conf/wg/KralKTW01},
so that assuming $\textsf{P}\neq\textsf{NP}$ we cannot hope to solve $k$-\Color{} in graphs with $\bar{K}_4$-free modulators of size $p$ in $f(p)\poly(n)$ time, for any function $f(\cdot)$.
As mentioned in the \Cref{sec:intro}, $(n-k)$-\SetCover{} can be solved in $O^*(2^{3k/2})$ time using a different method.
Can this approach be combined with our triangle packing argument to obtain a faster \DualColor{} algorithm? 

Another interesting research direction is to obtain faster algorithms for \emph{multiplicative} below-guarantee graph coloring.
For any parameter $\delta \in [1/3,1]$, by setting $k = (1-\delta)n$ in \Cref{thm:dualcolor} we get that $(\delta n)$-\Color{} can be solved in $O^*(2^{(1 - (\delta - 1/3)n})$ time.  
Using very different techniques,
previous work showed that for $\delta\in [0,1]$ it is possible to 
distinguish 
between the cases where the input graph $G$ admits a proper coloring with $(\delta n - 1)$ colors and the case where any proper coloring of $G$ needs at least $(\delta n + 1)$ colors 
in $O^*(2^{(1-\Omega(\delta^4))n})$ time \cite[Theorem 1.2]{Nederlof2016}. 
Can the dependence on $\delta$ be improved and the need for additive approximation here be removed,  to solve $(\delta n)$-\Color{} in $O^*(2^{(1-\Omega(\delta))n})$ time in general? 

Finally, we showed that for \DualColor{} and its harder variant, constructing a maximal triangle packing instead of a maximal matching can accelerate algorithms for these problems.
Can this method be used more generally to help improve parameterized algorithms or kernelization for other problems where the current best methods rely on matching-based techniques such as crown decomposition?

\bibliographystyle{plain}
\bibliography{main}

\begin{thebibliography}{10}

\bibitem{DBLP:conf/stacs/AkmalK25}
Shyan Akmal and Tomohiro Koana.
\newblock Faster edge coloring by partition sieving.
\newblock In {\em Proceedings of the 42nd International Symposium on Theoretical Aspects of Computer Science, {STACS} 2025}, volume 327 of {\em LIPIcs}, pages 7:1--7:18. Schloss Dagstuhl - Leibniz-Zentrum f{\"{u}}r Informatik, 2025.

\bibitem{AlferovBB24}
Vasily Alferov, Ivan Bliznets, and Kirill Brilliantov.
\newblock {Parameterization of (Partial) Maximum Satisfiability above Matching in a Variable-Clause Graph}.
\newblock In {\em Proceedings of the 38th {AAAI} Conference on Artificial Intelligence, {AAAI} 2024}, pages 7918--7925. {AAAI} Press, 2024.

\bibitem{DBLP:journals/ijm/AppelH77}
Kenneth Appel and Wolfgang Haken.
\newblock {Every planar map is four colorable. I. Discharging}.
\newblock {\em Illinois J. Math.}, 21(3):429--490, 1977.

\bibitem{subset-convolution}
Andreas Bj\"{o}rklund, Thore Husfeldt, Petteri Kaski, and Mikko Koivisto.
\newblock {Fourier meets M\"{o}bius: fast subset convolution}.
\newblock In {\em Proceedings of the 39th annual ACM symposium on Theory of computing, STOC 2007}. ACM, June 2007.

\bibitem{BHK2009}
Andreas Bj\"{o}rklund, Thore Husfeldt, and Mikko Koivisto.
\newblock {Set Partitioning via Inclusion-Exclusion}.
\newblock {\em SIAM Journal on Computing}, 39(2):546–563, January 2009.

\bibitem{DBLP:journals/ejc/BollobasCE80}
B{\'{e}}la Bollob{\'{a}}s, Paul~A. Catlin, and Paul Erd{\"{o}}s.
\newblock {Hadwiger's Conjecture is True for Almost Every Graph}.
\newblock {\em Eur. J. Comb.}, 1(3):195--199, 1980.

\bibitem{ChorFJ04}
Benny Chor, Mike Fellows, and David~W. Juedes.
\newblock {Linear Kernels in Linear Time, or How to Save $k$ Colors in $O(n^2)$ Steps}.
\newblock In {\em Proceedings of the 30th International Workshop on Graph-Theoretic Concepts in Computer Science}, volume 3353 of {\em WG 2004}, pages 257--269. Springer, 2004.

\bibitem{DBLP:journals/anm/ChudnovskyRST06}
Maria Chudnovsky, Neil Robertson, Paul~D. Seymour, and Robin Thomas.
\newblock The strong perfect graph theorem.
\newblock {\em Ann. Math.}, 164(1):51--229, 2006.

\bibitem{ParameterizedTextbook}
Marek Cygan, Fedor~V. Fomin, Łukasz Kowalik, Daniel Lokshtanov, Dániel Marx, Marcin Pilipczuk, Michał Pilipczuk, and Saket Saurabh.
\newblock {\em Parameterized Algorithms}.
\newblock Springer International Publishing, 2015.

\bibitem{KernelizationTextbook}
Fedor~V. Fomin, Daniel Lokshtanov, Saket Saurabh, and Meirav Zehavi.
\newblock {\em Kernelization: Theory of Parameterized Preprocessing}.
\newblock Cambridge University Press, December 2018.

\bibitem{GJS1976}
M.R. Garey, D.S. Johnson, and L.~Stockmeyer.
\newblock {Some simplified {\NP}-complete graph problems}.
\newblock {\em Theoretical Computer Science}, 1(3):237–267, February 1976.

\bibitem{GMOW2021}
Gregory Gutin, Diptapriyo Majumdar, Sebastian Ordyniak, and Magnus Wahlstr\"{o}m.
\newblock {Parameterized Pre-Coloring Extension and List Coloring Problems}.
\newblock {\em SIAM Journal on Discrete Mathematics}, 35(1):575–596, January 2021.

\bibitem{GutinMnichSurvey}
Gregory Gutin and Matthias Mnich.
\newblock {A Survey on Graph Problems Parameterized Above and Below Guaranteed Values}, 2024.

\bibitem{DBLP:books/wi/JensenT11}
Tommy~R. Jensen and Bjarne Toft.
\newblock {\em Graph Coloring Problems}.
\newblock Wiley, 2011.

\bibitem{Koa23}
Tomohiro Koana.
\newblock Induced matching below guarantees: Average paves the way for fixed-parameter tractability.
\newblock In {\em Proceedings of the 40th International Symposium on Theoretical Aspects of Computer Science, {STACS} 2023}, volume 254 of {\em LIPIcs}, pages 39:1--39:21. Schloss Dagstuhl - Leibniz-Zentrum f{\"{u}}r Informatik, 2023.

\bibitem{DBLP:conf/wg/KralKTW01}
Daniel Kr{\'{a}}l, Jan Kratochv{\'{\i}}l, Zsolt Tuza, and Gerhard~J. Woeginger.
\newblock Complexity of coloring graphs without forbidden induced subgraphs.
\newblock In {\em Proceedings of the 27th International Workshop on Graph-Theoretic Concepts in Computer Science, WG 2001}, volume 2204 of {\em Lecture Notes in Computer Science}, pages 254--262. Springer, 2001.

\bibitem{LDYR2021}
Wenjun Li, Yang Ding, Yongjie Yang, and Guozhen Rong.
\newblock {A {$(2+\eps)k$}-vertex kernel for the dual coloring problem}.
\newblock {\em Theoretical Computer Science}, 868:6–11, May 2021.

\bibitem{Lovasz1973}
L{\'a}szl{\'o} Lov{\'a}sz.
\newblock Coverings and colorings of hypergraphs.
\newblock {\em Proceedings of the 4th Southeastern Conference of Combinatorics, Graph Theory, and Computing}, pages 3--12, 1973.

\bibitem{Motwani1995}
Rajeev Motwani and Prabhakar Raghavan.
\newblock {\em Randomized Algorithms}.
\newblock Cambridge University Press, August 1995.

\bibitem{murota1999matrices}
Kazuo Murota.
\newblock {\em Matrices and matroids for systems analysis}, volume~20.
\newblock Springer Science \& Business Media, 1999.

\bibitem{Nederlof2016}
Jesper Nederlof.
\newblock {Finding Large Set Covers Faster via the Representation Method}.
\newblock Schloss Dagstuhl – Leibniz-Zentrum f\"{u}r Informatik, 2016.

\bibitem{Rote2001}
G\"{u}nter Rote.
\newblock {\em {Division-Free Algorithms for the Determinant and the Pfaffian: Algebraic and Combinatorial Approaches}}, page 119–135.
\newblock Springer Berlin Heidelberg, 2001.

\end{thebibliography}

\end{document}